\documentclass[conference]{IEEEtran}
\ifCLASSINFOpdf
\else
\fi
\usepackage{amsmath, amsfonts, amsthm, amssymb, wasysym, txfonts, mathrsfs,hyperref}

\usepackage[linesnumbered,ruled,vlined]{algorithm2e}

\usepackage{pdfpages}

\usepackage{tikz}

\newtheorem{thm}{Theorem}[section]
\newtheorem{co}[thm]{Corollary}
\newtheorem{lem}[thm]{Lemma}

\newtheorem{assumption}[thm]{Assumption}

\newtheorem{definition}[thm]{Definition}

\newtheorem{example}[thm]{Example}

\newtheorem{remark}[thm]{Remark}

\newcommand{\bracenom}{\genfrac{\lbrace}{\rbrace}{0pt}{}}
\usepackage{algorithmic}
\begin{document}

\title{SASH: Decoding Community \\Structure in Graphs}

\author{
\IEEEauthorblockN{Allison Beemer\IEEEauthorrefmark{1}, Jessalyn Bolkema\IEEEauthorrefmark{2}}
\IEEEauthorblockA{\IEEEauthorrefmark{1}Department of Mathematics, University of Wisconsin-Eau Claire,
Eau Claire, WI 54701
}
\IEEEauthorblockA{\IEEEauthorrefmark{2} Department of Mathematics, California State University, Dominguez Hills, Carson, CA 90747 }}

\maketitle

\begin{abstract} 
Detection of communities in a graph entails identifying clusters of densely connected vertices; the area has a variety of important applications and a rich literature. The problem has previously been situated in the realm of error correcting codes by viewing a graph as a noisy version of the assumed underlying communities. In this paper, we introduce an encoding of community structure along with the resulting code’s parameters. We then present a novel algorithm, SASH, to decode to estimated communities given an observed dataset. We demonstrate the performance of SASH via simulations on an assortative planted partition model and on the Zachary's Karate Club dataset.
\end{abstract}

\section{Introduction}

Community detection is the problem of identifying clusters within a network: subsets of nodes that are more densely connected to each other than they are to the rest of the graph. Identifying this underlying structure is a fundamental question in network science with wide ranging applicability, including in biology, social sciences, cybersecurity, and beyond \cite{javed2018community}. 
Many clustering methods exist, each with their own strengths and limitations \cite{fortunato2010community}; unfortunately, there is no grand unified theory of community detection, nor is there likely to be \cite{aldecoa2013exploring}, given the wide variation in network properties, scale, and complexity. 
One branch of recent research seeks to understand the fundamental possibilities and limitations of community detection algorithms, identifying thresholds beyond which clustering is impossible or improbable \cite{abbe2018community}. 
% \jess{double check that there aren't better citations for this point. probably there are additional citations, at a minimum!}

The connection between the challenge of community detection and the tools of error correction is straightforward: if we think of the idealized community structure as an encoded message, the observed graph can be interpreted as a noisy version of that structure, distorted by missing or spurious relationships. In fact, this parallel is natural enough that existing literature already supports it. The connection was discussed in the context of information-theoretic bounds in \cite{abbe2015community}. In \cite{R18} and \cite{BR19}, a two-community detection model is described explicitly as a channel decoding problem, with the allowed clusterings expressed in the form of a linear code admitting graph-based message-passing decoding (specifically, Gallager). In this work, we propose a new model of community decoding by introducing a novel formulation of an error-correcting community code, as well as a decoding algorithm that exploits the unique properties of our community code. Our algorithm differs from previous work in its tailoring to the problem, as well as its flexibility in terms of the number and size of output community clusters. This work comprises a new technique for graph clustering and a new perspective for technical analysis.

% How is it different from other strategies? Talk about other ``code" approach introduced in \cite{abbe2015community} and developed in \cite{R18, BR19}. \jess{I think this is the story but need to double check} They use Gallager, make a linear code, split into two (allow the algorithm to be applied iteratively to split into more communities). In this model, both vertex labels and edge existence are information bits in the associated code; in our approach, we look only at edge existence or nonexistence, making our model inherently more flexible.

% Rather than prescribing the number of cliques, we place restriction on the smallest cluster size. The two are related, though not exactly...

Necessary background is discussed in Section \ref{sec:prelims}, and in Section \ref{sec:params} we give the parameters of the community code $C_{n,m}$. In Section \ref{sec:decoder}, we present our decoder, SASH. Section \ref{sec:simulations} contains simulation results for SASH, and Section \ref{sec:conc} concludes the paper.

\section{Preliminaries}
\label{sec:prelims}

As discussed above, the graph clustering problem can be phrased as a coding theory problem by viewing an existing dataset's graph as a noisy version of a clustering of the graph's vertices into disjoint cliques. In contrast to traditional channel coding problems, we are not permitted to freely design the code, but instead focus on efficient and accurate decoding of the received graph to the most likely codeword representing an allowed graph clustering.

We therefore begin by defining a family of \textit{community codes}, which determine the clusterings of a graph that can be output from our decoder. Throughout the paper, we will move back and forth between talking about a graph and the upper half of its corresponding \textit{adjacency matrix}. An adjacency matrix $A$ for a simple undirected graph $G$ with $n$ labeled vertices is a symmetric matrix such that a $1$ appears in entry $A_{ij}$, $i\neq j$, if and only if vertices $i$ and $j$ are adjacent in $G$. Otherwise, $A_{ij}=0$. 

\begin{definition}
       Let $n\geq 2$, $N=\binom{n}{2}$, and $1\leq m \leq n$. The community code $C_{n,m}\subseteq \mathbb{F}_{2}^{N}$ contains exactly those words corresponding to the upper half of the adjacency matrix of a (simple, undirected) graph on $n$ vertices consisting of disjoint cliques with minimum clique size lower-bounded by $m$. 
\end{definition}
% \jess{add example with graphs here}
\begin{example}
    The code $C_{4,1}$ consists of all simple graphs on four vertices, with no restrictions on minimum clique size. In comparison, the code $C_{4,2}$ has minimum clique size $m=2$ and thus contains only those graphs on four vertices with no isolated vertices. All codewords are shown in Figure~\ref{fig:graphs}.

\end{example}

    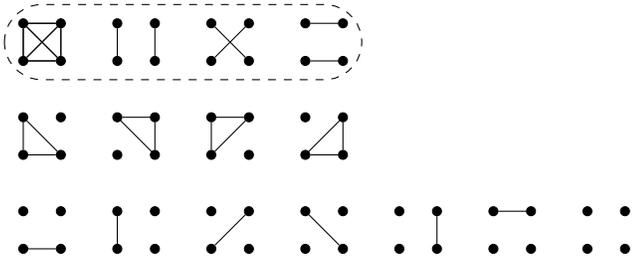
\begin{figure}
    \begin{tikzpicture}[scale=0.5, every node/.style={circle, draw, fill=black, inner sep=1.2pt}]
            % K_4 complete graph
        \node (A) at (0,5) {};
        \node (B) at (1,5) {};
        \node (C) at (0,6) {};
        \node (D) at (1,6) {};
        \foreach \i in {A,B,C,D} {
            \foreach \j in {A,B,C,D} {
                \ifthenelse{\equal{\i}{\j}}{}{\draw (\i) -- (\j);}
            }
        }
        
        % Empty graph
        \node (A1) at (15,0) {};
        \node (B1) at (16,0) {};
        \node (C1) at (15,1) {};
        \node (D1) at (16,1) {};

            % Two disjoint edges
        \node (A) at (7.5,5) {};
        \node (B) at (8.5,5) {};
        \node (C) at (7.5,6) {};
        \node (D) at (8.5,6) {};
        \draw (A) -- (B);
        \draw (C) -- (D);
    
        % Two disjoint edges
        \node (A) at (2.5,5) {};
        \node (B) at (3.5,5) {};
        \node (C) at (2.5,6) {};
        \node (D) at (3.5,6) {};
        \draw (A) -- (C);
        \draw (B) -- (D);
    
        % Two disjoint edges
        \node (A) at (5,5) {};
        \node (B) at (6,5) {};
        \node (C) at (5,6) {};
        \node (D) at (6,6) {};
        \draw (A) -- (D);
        \draw (B) -- (C);

        % Triangle and isolated vertex
        \node (A) at (0,2.5) {};
        \node (B) at (1,2.5) {};
        \node (C) at (0,3.5) {};
        \node (D) at (1,3.5) {};
        \draw (A) -- (B);
        \draw (B) -- (C);
        \draw (C) -- (A);
  
          % Triangle and isolated vertex
        \node (A) at (2.5,2.5) {};
        \node (B) at (3.5,2.5) {};
        \node (C) at (2.5,3.5) {};
        \node (D) at (3.5,3.5) {};
        \draw (D) -- (B);
        \draw (B) -- (C);
        \draw (C) -- (D);
    
        % Triangle and isolated vertex
        \node (A) at (5,2.5) {};
        \node (B) at (6,2.5) {};
        \node (C) at (5,3.5) {};
        \node (D) at (6,3.5) {};
        \draw (A) -- (D);
        \draw (D) -- (C);
        \draw (C) -- (A);
    
        % Triangle and isolated vertex
        \node (A) at (7.5,2.5) {};
        \node (B) at (8.5,2.5) {};
        \node (C) at (7.5,3.5) {};
        \node (D) at (8.5,3.5) {};
        \draw (A) -- (B);
        \draw (B) -- (D);
        \draw (D) -- (A);
    
        % One edge
        \node (A) at (0,0) {};
        \node (B) at (1,0) {};
        \node (C) at (0,1) {};
        \node (D) at (1,1) {};
        \draw (A) -- (B);
 
           % One edge
        \node (A) at (2.5,0) {};
        \node (B) at (3.5,0) {};
        \node (C) at (2.5,1) {};
        \node (D) at (3.5,1) {};
        \draw (A) -- (C);
    
        % One edge
        \node (A) at (5,0) {};
        \node (B) at (6,0) {};
        \node (C) at (5,1) {};
        \node (D) at (6,1) {};
        \draw (A) -- (D);
    
        % One edge
        \node (A) at (7.5,0) {};
        \node (B) at (8.5,0) {};
        \node (C) at (7.5,1) {};
        \node (D) at (8.5,1) {};
        \draw (B) -- (C);
    
        % One edge
        \node (A) at (10,0) {};
        \node (B) at (11,0) {};
        \node (C) at (10,1) {};
        \node (D) at (11,1) {};
        \draw (B) -- (D);
    
        % One edge
        \node (A) at (12.5,0) {};
        \node (B) at (13.5,0) {};
        \node (C) at (12.5,1) {};
        \node (D) at (13.5,1) {};
        \draw (C) -- (D);
        % Draw the rounded rectangle
        \draw[rounded corners=5mm, dashed] (-0.5, 4.5) rectangle (9, 6.5);
         % Draw the rounded rectangle
         % \draw[rounded corners=5mm, draw=black, dashed] (-1, -3.5) rectangle (14.5, 9.5);
    \end{tikzpicture}
    \caption{All codewords in $C_{4,1}$, with $C_{4,2}$ indicated by the dashed region.}
    \label{fig:graphs}
    \end{figure}
    
% \begin{tikzpicture}

%  \end{tikzpicture}
As suggested in the above definition, we will use $n$ to denote the number of vertices in our graph, and $m$ to indicate the minimum clique (cluster) size. The block length of the community code, $\binom{n}{2}$, will be denoted by $N$. We will discuss the parameters of ${C}_{n,m}$ further in Section \ref{sec:params}.

In real-world applications, there is not necessarily one correct clustering of a given dataset. Instead, one way to measure the accuracy of a particular clustering algorithm (in our case, decoder) is to fabricate a noisy set of clusters using the so-called \emph{planted partition} model \cite{holland1983stochastic}. %(stochastic block model? what's the best term here? Planted partition is fine for what we are doing! SBM is the most general term). 
In such a model, cluster labels are predetermined, and a graph is constructed on the set of labeled vertices by adding an edge within a cluster with one probability, $P$, and adding an edge between clusters with a different probability, $Q$. 

Indeed, the process of constructing the planted partition graph is equivalent to sending the codeword corresponding to a set of disjoint cliques through a binary asymmetric channel (BAC).
Let the probability that we see an edge where there is no edge in the underlying community structure be denoted $p$, and let the probability that we see no edge where there is one in the community structure be denoted $q$. 
We assume %(without loss of generality?? \jess{yes, I think this is the standard \emph{assortative} model assumption! Maybe actually $p\leq 1-q$})
that $p \leq q$ and $0<p+q<1$. In particular, it is more likely that an edge is missing where there should be one than that an extra edge has appeared. Notice that $p=Q$ from the planted partition model, and $q=1-P$. Since $p+q<1$, we recover $Q<P$, which distinguishes an \textit{assortative} planted partition model.

% [does this match community detection literature? maybe - the typical phrasing would be that we are building a random graph on labeled vertices by adding vertices within a part with probability $P$ and adding edges between distinct parts with probability $Q$, so the assortative model assumes $P>Q$. Our channel model removes edges within a part with probability $q$, so $q=1-P$, and , so $q=1-P<1-Q=1-p$, or $p+q< 1$? Channel model adds edges between parts with probability $p$, so $Q=p$? I think I have this in some slides somewhere so let me double-check]. 

Because our ``channel'' is asymmetric, we refer to the work of Cotardo and Ravagnani \cite{CR22}, which introduces \textit{discrepancy} as an extension of the familiar Hamming distance in the following sense:

\begin{thm}\cite{CR22}
Maximum likelihood decoding over a BAC corresponds to nearest neighbor decoding with ``distance'' between words given by discrepancy.
\end{thm}

The discrepancy between two codewords corresponds to the Hamming distance exactly when the crossover probabilities of a binary bit-flip channel are equal: i.e., when the channel is symmetric. In the asymmetric case, a larger weight is given to the less likely flip:

\begin{definition}\cite{CR22}
    {Assume $p\leq q$ and $0<p+q<1$}. For $\mathbf{x},\mathbf{y} \in \mathbb{F}_2^n$, and $a,b\in\mathbb{F}_2$, let 
    $d_{ab}(\mathbf{y},\mathbf{x}) = |\{i \mid y_i=a, x_i=b\}|.$
    The \textit{discrepancy} between $\mathbf{y}$ and $\mathbf{x}$ is then given by
    \begin{equation}
        \delta_{p,q}(\mathbf{y},\mathbf{x})=\gamma_{p,q}d_{10}(\mathbf{y},\mathbf{x})+d_{01}(\mathbf{y},\mathbf{x})
    \end{equation}
    where
    \begin{equation}
    \label{eq:gamma_def} 
    \gamma_{p,q} = \log_{\frac{q}{1-p}}\left(\frac{p}{1-q}\right)
    \end{equation}
\end{definition}

%{\color{red}\href{https://www.desmos.com/3d/opjw2zugfw}{here's a graph of gamma}}

When the channel is understood, we will drop the subscripts on $\gamma$ and $\delta$. Observe that discrepancy is not symmetric, so is not a metric. We also make use of the following result:
\begin{lem}\cite{CR22}
    For $p\leq q$ and $0<p+q<1$, we have $\gamma_{p,q} \geq 1$, with equality when $p=q$.
\end{lem}

\begin{remark}
    The authors of \cite{CR22} assume $0<p\leq q <0.5$, but all the results we rely on here are true with $p\leq q$ and $0<p+q<1$. We prefer the latter condition because it allows for cliques to become significantly less dense in the observed dataset relative to the idealized community cluster. More generally, while a binary symmetric channel has a symmetry that easily allows for a restriction on crossover probability without loss of generality, this is not the case for the asymmetric case. We refer the reader to e.g. \cite{judy}, where the authors make the same assumptions we do.
\end{remark}

In the remainder of the paper, we let $[n]:=\{1,\ldots,n\}$. We indicate that two vertices are adjacent in an undirected graph using the notation $x\sim y$ (equivalently, $y\sim x$). The Hamming distance between two vectors is denoted $d_{H}(\cdot,\cdot)$, and the (Hamming) weight of a vector is denoted $w_{H}(\cdot)$. The weight distribution of a code with block length $N$ is given by a vector of length $N+1$, where nonnegative integer entry $i$ is the number of codewords of weight $i-1$.

\section{Community Code Parameters}
\label{sec:params}

In this section, we discuss the size, rate, and minimum discrepancy of the $C_{n,m}$ community codes. We begin by observing the following:

\begin{lem}
For integers $n\geq 3$, $m\in [n]$, $C_{n,m}$ is nonlinear.
\end{lem}

\begin{proof}
We will view codewords as their corresponding graphs on $n$ vertices. Note that the all-zero codeword corresponds to the graph on $n$ vertices with no edges. The sum of two words in binary vector representation corresponds to the graph on $n$ vertices such that $x\sim y$ if and only if $x\sim y$ in exactly one of the summands.
First, assume $n\geq 3$, and $m>1$. Then, the all-zero word is not a codeword, and $C_{n,m}$ is not linear.
Otherwise, if $m=1$, consider a codeword graph $G$ that has a clique of size $3$ on vertices $x,y,z$. Let $G'$ be a graph such that $x,y$ form a clique. The sum of $G$ and $G'$ contains a path of length 2, which is not permissible in a codeword. Thus, $C_{n,m}$ is nonlinear.
\end{proof}

Next, we turn to the size of $C_{n,m}$. In the case where communities of size one are allowed, we have the following: 

\begin{lem}
\label{lem:size_m1}
    % When $m=1$, $n\geq 1$, $|C_{n,1}|$ is equal to the $n$th \textit{Bell number}, which may be defined recursively for $n\geq 1$ by
    % \begin{equation}
    %    B_{n+1} = \sum_{k=0}^{n}\binom{n}{k}B_k 
    % \end{equation}
    % for $n\geq 0$, with $B_0=1$, or explicitly for $n\geq 1$ as:
    When $m=1$, $n\geq 1$, $|C_{n,1}|$ is equal to the $n$th \textit{Bell number}, which may be defined recursively for $n\geq 1$ by
    \begin{equation}
       B_{n} = \sum_{k=0}^{n-1}\binom{n-1}{k}B_k, 
    \end{equation}
    with $B_0=1$, or explicitly for $n\geq 1$ as:
    % \begin{equation}
    %     B_n = \sum_{k=0}^{n}\bracenom{n}{k} = \sum_{k=0}^{n}\sum_{i=0}^{k}\frac{(-1)^{k-i}i^n}{(k-i)!i!},
    % \end{equation}
    \begin{equation}
        B_n = \sum_{k=1}^{n}\bracenom{n}{k} = \sum_{k=1}^{n}\sum_{i=0}^{k}\frac{(-1)^{k-i}i^n}{(k-i)!i!},
    \end{equation}
    where $\bracenom{n}{k}$ indicates the \textit{Stirling number of the second kind}.
\end{lem}
\begin{proof}
    The Bell numbers count the number of partitions of $[n]$. The possible community structures on a graph with $n$ vertices, allowing for communities of size one, is the set of all partitions of the $n$ vertices.
\end{proof}

Recall that the Stirling number of the second kind $\bracenom{n}{k}$ counts the number of partitions of $[n]$ with exactly $k$ parts. Relevant to our work here, the \textit{$m$-associated Stirling number of the second kind} $S_{m}(n,k)$ counts the number of partitions of $[n]$ with exactly $k$ parts such that each part contains at least $m$ elements. This allows us to generalize Lemma \ref{lem:size_m1} to the case where $m>1$:

\begin{lem}
When $m>1$, $n\geq m$, 
% \begin{equation}
%     |C_{n,m}|=B_m(n) := \sum_{k=0}^{n}S_{m}(n,k)
% \end{equation}
\begin{equation}
    |C_{n,m}|=B_m(n) := \sum_{k=1}^{n}S_{m}(n,k)
\end{equation}
where $S_{m}(n,k)$ indicates the $m$-associated Stirling number of the second kind. 
%Recursively for $n\geq m-1$,
%     \begin{equation}
% B_m(n+1)=\sum_{k=0}^{n+1-m} \binom{n}{k}B_m(k)
%     \end{equation}
%     for $n\geq 0$, with $B_m(0),\ldots, B_{m}(m-2)=0$, $B_m(m-1)=1$.
Recursively for $n\geq m$,
    \begin{equation}
B_m(n)=\sum_{k=0}^{n-m} \binom{n-1}{k}B_m(k),
    \end{equation}
with $B_m(0),\ldots, B_{m}(m-2)=0$, $B_m(m-1)=1$.
\end{lem}
\begin{proof}
    % Consider a partition of $[n+1]$ into parts of size at least $m$. Without loss of generality, assume the element $(n+1)$ is in the first part of the partition, along with $j$ other elements, where $m-1\leq j \leq n$. There are $\binom{n}{j}$ ways to choose these $j$ elements and $B_m(n-j)$ ways to partition the remaining elements into parts of size at least $m$. Thus 
    % \begin{align} 
    %     B_m(n+1) & =\sum_{j=m-1}^{n} \binom{n}{j}B_m(n-j) \\
    %     % & = \sum_{j=m-1}^{n} \binom{n}{n-j}B_m(n-j) \\
    %     & = \sum_{k=0}^{n-(m-1)} \binom{n}{k}B_m(k) .
    % \end{align}
    Consider a partition of $[n]$ into parts of size at least $m$. Without loss of generality, assume the element $n$ is in the first part of the partition, along with $j$ other elements, where $m-1\leq j \leq n-1$. There are $\binom{n-1}{j}$ ways to choose these $j$ elements and $B_m(n-1-j)$ ways to partition the remaining elements into parts of size at least $m$. Thus 
    \begin{align} 
        B_m(n) & =\sum_{j=m-1}^{n-1} \binom{n-1}{j}B_m(n-1-j) \\
        & = \sum_{j=m-1}^{n-1} \binom{n-1}{n-1-j}B_m(n-1-j) \\
        & = \sum_{k=0}^{n-1-(m-1)} \binom{n-1}{k}B_m(k) .
    \end{align}
\end{proof}
\vspace{-0.2in}
% \begin{remark}
%     We note that we have...
%     Another way to think Stirling is a sum over partition types of multinomial coefficient per partition type \allison{sort of...?}. Given the type, it is straightforward to compute the weight of the codeword. This will become relevant in decoding; see Section~\ref{sec:decode}.
% \end{remark}

\begin{example}

Table \ref{tab:ex_wts_size} gives an example of the possible partition types, their codeword weights, and all $m$-associated Stirling numbers of the second kind (corresponding to choices of minimum clique size) for varying numbers of parts ($k$) in a partition of $n=4$. See Figure \ref{fig:graphs} for a graphical representation.

\begin{table}[hbt!]
    \centering
\begin{tabular}{|c||c|cc|c|c|}
\hline
$k$             & 1 & \multicolumn{2}{c|}{2}         & 3     & 4       \\ \hline  \hline
partition types & 4 & \multicolumn{1}{c|}{1,3} & 2,2 & 1,1,2 & 1,1,1,1 \\ \hline
weight          & 6 & \multicolumn{1}{c|}{3}   & 2   & 1     & 0       \\ \hline
$S_{1}(4,k)=\bracenom{4}{k}$    & 1 & \multicolumn{2}{c|}{7}         &   6   &    1   \\ \hline
$S_{2}(4,k)$             & 1 & \multicolumn{2}{c|}{3}         &  0   &    0   \\ \hline
$S_{3}(4,k)=S_{4}(4,k)$             & 1 & \multicolumn{2}{c|}{0}         &   0  &    0  \\ \hline

\end{tabular}
\vspace{0.1in}
\caption{Partition type, weight, and ($m$-associated)\\ Stirling numbers of the second kind when $n=4$.}
\label{tab:ex_wts_size}
\end{table}
When $m=1$, adding up the ($1$-associated) Stirling numbers of the second kind, we see 
$|C_{4,1}|=\sum_{k=1}^4\bracenom{4}{k}=15.$
The code $C_{4,1}$ has weight distribution vector $(1,6,3,4,0,0,1)$. If $m$ increases, the code size decreases because the set of allowed partition types is a subset of the previously-allowed types. For example, if $m=2$, only partition types $4$ and $2,2$ are allowed. This gives $|C_{4,2}|=\sum_{k=1}^4 S_2(4,k)=4$, with weight distribution vector $(0,0,3,0,0,0,1)$.
\end{example}

\begin{remark}
    Given the partition type of a codeword, it is straightforward to calculate its weight via a sum of binomial coefficients. However, it is possible for two distinct allowed partition types to have the same codeword weight. For example, the partitions $6,6$ and $2,2,8$ of $n=12$ both give codewords of weight $30$. The ease of weight computation and the non-injectivity of the weight calculation will come into play in the execution of SASH; see Section \ref{sec:decoder}.
\end{remark}

We next present the asymptotic rate of $C_{n,m}$ for mathematical interest. However, we observe that rate is not relevant in this scenario in the usual way (e.g. it is not necessarily desirable that the rate stay bounded away from zero), as we are only receiving, not transmitting, information.

\begin{lem}
\label{lem:rate_to_0}
Consider the family of $C_{n,m}$ codes for a fixed integer $m$. The rate of $C_{n,m}$ approaches zero as $n\to \infty$.

\end{lem}
\begin{proof}
    Since $|C_{n,m}|\geq |C_{n,m+1}|$, we will prove the statement only for $m=1$. In this case, the rate is given by 
    \begin{align}
        \frac{\log_{2}(B_n)}{N} &= \frac{\log_{2}\left(\sum_{k=1}^{n}\sum_{i=0}^{k}\frac{(-1)^{k-i}i^n}{(k-i)!i!}\right)}{N} \\
    &\leq \frac{\log_{2}\left(\binom{n+2}{2} \cdot n^ n\right)}{\binom{n}{2}} \\
    %pull out the n^n, make the alternating all positive
    % then there are 2+3+...+(n+1) terms, all at most 1
    &= \frac{2\log_{2}\left((n+2)(n+1)\right)-2}{n(n-1)} +\frac{2\log_{2}\left(n\right)}{(n-1)}\label{eq:rate_ub}
    \end{align}
    Equation \eqref{eq:rate_ub} approaches 0 as $n\to \infty$.
\end{proof}

\begin{co}
    Consider the family of $C_{n,m}$ codes for $m=\lceil \frac{n}{k}\rceil $, where $k$ is a fixed positive real number. The rate of $C_{n,m}$ approaches zero as $n\to \infty$.  
\end{co}
\begin{proof}
    This is easily shown using Lemma \ref{lem:rate_to_0}  and the fact that $|C_{n,\frac{n}{k}}| \leq |C_{n,1}|$.
\end{proof}

% \begin{itemize}
%     % \item  Otherwise... \href{https://math.stackexchange.com/questions/515051/bell-number-with-minimum-bound-on-partition-size}{This answer} 
%     \item Rate (goes down with $m$) will be a trade-off with decoding complexity, potentially? But increasing $m$ could potentially increase accuracy, if small cliques don't make sense in context.
%     \item to-do: Prove that rate goes to 0 (we are confident of this but should write this down!)
%     \item to-do: maybe talk about generating functions/prove exponential generating function for this nice sequence 
% \end{itemize}

Finally, we present results on minimum discrepancy. We focus on $n\geq 2m$, since if $n$ is less than $2m$, there is only one codeword in $C_{n,m}$.

\begin{thm}
\label{thm:disc_lower_bd}

    Assume $n\geq 2m$. The minimum discrepancy of $C_{n,m}$ is bounded below by 
    $\delta(C_{n,m})\geq \min\{m^2, (1+\gamma)m\}$,
    where $\gamma$ depends on the channel parameters and is as defined in Equation \eqref{eq:gamma_def}.
    %May need $m\geq 3m$ so that we don't have a single clique that must be split into e.g. $m$ and $m+7$. Hmm how annoying. Maybe it's okay because it's a lower bound?
\end{thm}

Our proof of the above result, which we omit here for the sake of space, argues that no discrepancy between two codewords of $C_{n,m}$ can be smaller than the given minimum by splitting into two cases: (1) each vertex only loses or gains edges, or (2) some vertex both loses and gains edges. The first case's argument shares logic with the proof of Theorem \ref{thm:disc_upper_bd} below. In the second case, we fan out from the assumed vertex that both loses and gains edges, arguing a ripple effect that results in at least the stated discrepancy.

We observe that the minimum is necessary, as the smaller value will be determined by the value of $\gamma$, which in turn is calculated based on BAC crossover probabilities $p$ and $q$.

\begin{thm}
\label{thm:disc_upper_bd}
    Assume $n\geq 3m+1$. The minimum discrepancy of $C_{n,m}$ is bounded above by $\delta(C_{n,m})\leq \min\{m^2,(1+\gamma)m\}$, where $\gamma$ depends on the channel parameters and is as defined in Equation \eqref{eq:gamma_def}.
\end{thm}
\begin{proof}
To prove the bound, we exhibit a pair of codewords with discrepancy $(1+\gamma)m$ and a pair with discrepancy $m^2$.

To achieve a discrepancy of $(1+\gamma)m$, consider the codeword that has three cliques of sizes $m+1$, $m$, and $n-(2m+1)$. Take one vertex of the clique of size $m+1$, remove it from this clique, and add it to the clique of size $m$. The result remains a codeword, and the discrepancy between the two words is equal to $(1+\gamma)m$, as desired.

To achieve a discrepancy of $m^2$, consider the codeword that has two cliques of sizes $2m$ and $n-2m$. Split the clique of size $2m$ into two, each of size $m$. The result remains a codeword, and the discrepancy between the two words is equal to $m^2$, as desired.
\end{proof}

The following is a combination of Theorems \ref{thm:disc_lower_bd} and \ref{thm:disc_upper_bd}.

\begin{co}
    For $n\geq 3m+1$, $\delta(C_{n,m})= \min\{m^2,(1+\gamma)m\}$.
\end{co}

% \begin{remark}
%     The missing exact values of discrepancy are for the range $2m\leq n \leq 3m$. Notice that if $n<2m$, there will be only one codeword in the code, which is the clique on $n$ vertices. For the missing range, the codewords can have up to 2 cliques for $n<3m$, and up to $3$ (balanced) cliques for $n=3m$.
% \end{remark}

Using similar arguments, we may also show the following:

\begin{thm}
    For $2m<n< 3m$, 
    \begin{equation}
        \delta(C_{n,m})= \min\{(n-m-1)+\gamma m,m(n-m)\}.
    \end{equation}
    At $n=2m,3m$, we have:
    \begin{align}
        \delta(C_{2m,m})&= \min\{2(m-1)(1+\gamma),m^2\},\\
        \delta(C_{3m,m})&= \min\{2(m-1)(1+\gamma),m^2, 2m-1+\gamma m\}.
    \end{align}
\end{thm}

Intuitively, the above results on minimum discrepancy suggest that the larger $m$ is, and the larger $\gamma$ is, the larger the minimum discrepancy between codewords, and the more accurate a decoder that minimizes discrepancy will be.

\begin{example}
\label{ex:min_disc}
    Over a BAC with $p=0.1$, $q=0.3$, $\delta(C_{16,4})=\min\{4(1+\gamma),16\}\approx 11.085$ and $\delta(C_{16,5})=\min\{5(1+\gamma),25\}\approx 13.856$.
    For these channel parameters, the only time the $m^2$ term is the minimum is $m=2$; in this case, $\delta(C_{16,2})=\min\{2(1+\gamma),4\}=4$.
    We see jumps at $m=6$ and $m=8$, with $\delta(C_{16,6})=\min\{9+6\gamma,60\}\approx 19.628$, and $\delta(C_{16,8})=2(m-1)+2\gamma(m-1)=14+14\gamma\approx 38.797$. Starting at $m=9$, the code has a single codeword.
\end{example}

\section{SASH Decoding Algorithm}
\label{sec:decoder}

Unsurprisingly, maximum likelihood (ML) decoding of $C_{n,m}$ is computationally infeasible. In this setting, we observe that ML is also equivalent to (multiple instances of) the graph alignment problem. Thus, we present a novel decoding algorithm customized to the community code's properties: Slide Along, Shuffle/Hop (SASH).

SASH makes use of the weight distribution of the community codes, as well as the fact that Maximum Likelihood decoding is equivalent to Nearest Neighbor decoding with discrepancy. Recall that $C_{n,m}$ has low rate in general, and that it is straightforward to construct codewords with a given weight. 

SASH moves radially out (Slide Along) from the Hamming weight of our received word, randomly checking a fixed number of codewords selected uniformly at random (Shuffle) from each partition type (Hop) with a particular weight. The closest encountered codeword, in terms of discrepancy, is returned at the end of this process.
Because the difference in weights between two words gives a lower bound on their Hamming distance, and hence discrepancy, once we find a codeword at discrepancy $\delta$ from our received word, this puts an upper bound on how far out (in terms of weight) from our received word we need to search. See Algorithm \ref{alg:decoder} for pseudocode of SASH.

\begin{algorithm} 
\caption{\label{alg:decoder}SASH}
\begin{algorithmic}[1]
\STATE \textbf{input:} $\mathbf{y}$, $n$, $m$, $p$, $q$, $t$
\STATE $w\gets wt_{H}(y)$; $L\gets$ $\{wt_{H}(\mathbf{c})\mid \mathbf{c}\in C_{n,m}\}$
\FOR {$W \in L$}
\STATE $P_{n,m}(W) \gets \{$allowed partition types with weight $W\}$
\ENDFOR
\STATE $\rho \gets \binom{n}{2}$; $\Delta \gets 0$
\WHILE {$\Delta <\rho$}
\FOR {$W \in \{w-\Delta,w+\Delta\}\cap L$}
\FOR {$p \in P_{n,m}(W)$} 
\STATE choose $\mathbf{c}_1,\ldots,\mathbf{c}_t$ of partition type $p$ 
\STATE $\delta, \ j = \min_{i}\delta_{p,q}
(\mathbf{y},\mathbf{c}_{i}), \ \text{argmin}_{i}\delta_{p,q}
(\mathbf{y},\mathbf{c}_{i})$
\IF {$\delta < \rho$}
\STATE $\hat{\mathbf{x}} \gets \mathbf{c}_j$; $\rho \gets \delta$
\ENDIF 
\ENDFOR
\ENDFOR
\STATE $\Delta \gets \Delta+1$
\ENDWHILE
\STATE \textbf{return:} $\hat{\mathbf{x}}$
\end{algorithmic}
\end{algorithm}

% \begin{remark}
%     Recall that the \textit{covering radius} of a code $C\subseteq \mathbb{F}_{q}^{N}$ is given by $\max_{\mathbf{x}\in \mathbb{F}_{q}^{n}}\min_{\mathbf{c}\in C}d_{H}(\mathbf{x},\mathbf{c})$. In other words, the covering radius is the smallest value $\rho$ such that balls of radius $\rho$ centered at codewords cover the entire space. Observe that before having knowledge of the first radius update in our algorithm, the initial radius of our decoder could be set to the covering radius of the community code, instead of $\binom{n}{2}$ as it currently is in Algorithm \ref{alg:decoder}. The covering radius of $C_{n,m}$ is a topic of ongoing work. If the covering radius of the code is $\rho$, note that we perform at most the following number of discrepancy calculations:
%     \[t
%     \cdot\max_{\mathbf{y}}\#\{ \text{partition types with weight } t \text{ s.t. } |t-wt_{H}(\mathbf{y})|\leq \rho\}\]
% \end{remark}

\section{Simulations}
\label{sec:simulations}

In this section, we present several simulations and performance measures to demonstrate the effectiveness of SASH.

The results in Figure \ref{fig:plot_accuracy} show the accuracy (over 100 trials at each $t$ value) of the decoder for a planted partition model on $n=16$ vertices. The clusters are randomly chosen according to a stipulation of $m=4, 6$, or $8$, and passed through an artificial BAC with $p=0.1$, $q=0.3$. We then run our decoder for the same $p,q$ values and with increasing $t$ values, checking whether the codeword returned after each trial is exactly the sent codeword, or whether it is ``as good'' in the sense that the discrepancy between the received word and codeword estimate is at most the discrepancy between the received word and sent codeword. Observe that overall, the performance of the decoder increases with increased $t$; recall that $t$ determines the number of codeword candidates that are checked at each partition type. It is also the case that increasing $m$ strengthens the performance of the decoder in the ``as good'' case. Recall from Example \ref{ex:min_disc} that the minimum discrepancy of $C_{16,m}$ jumps up at $m=6$, and again at $m=8$. However, we currently check each partition type a fixed number of times regardless of the number of possible types or the total number of codewords of each type. As a result, when $m=6$, more checks overall are performed than when $m=8$; we suspect this contributes to the unexpected relative ``exact'' performance for these $m$ values.  

\begin{figure}[hbt!]
    \centering
    \includegraphics[width=1\linewidth]{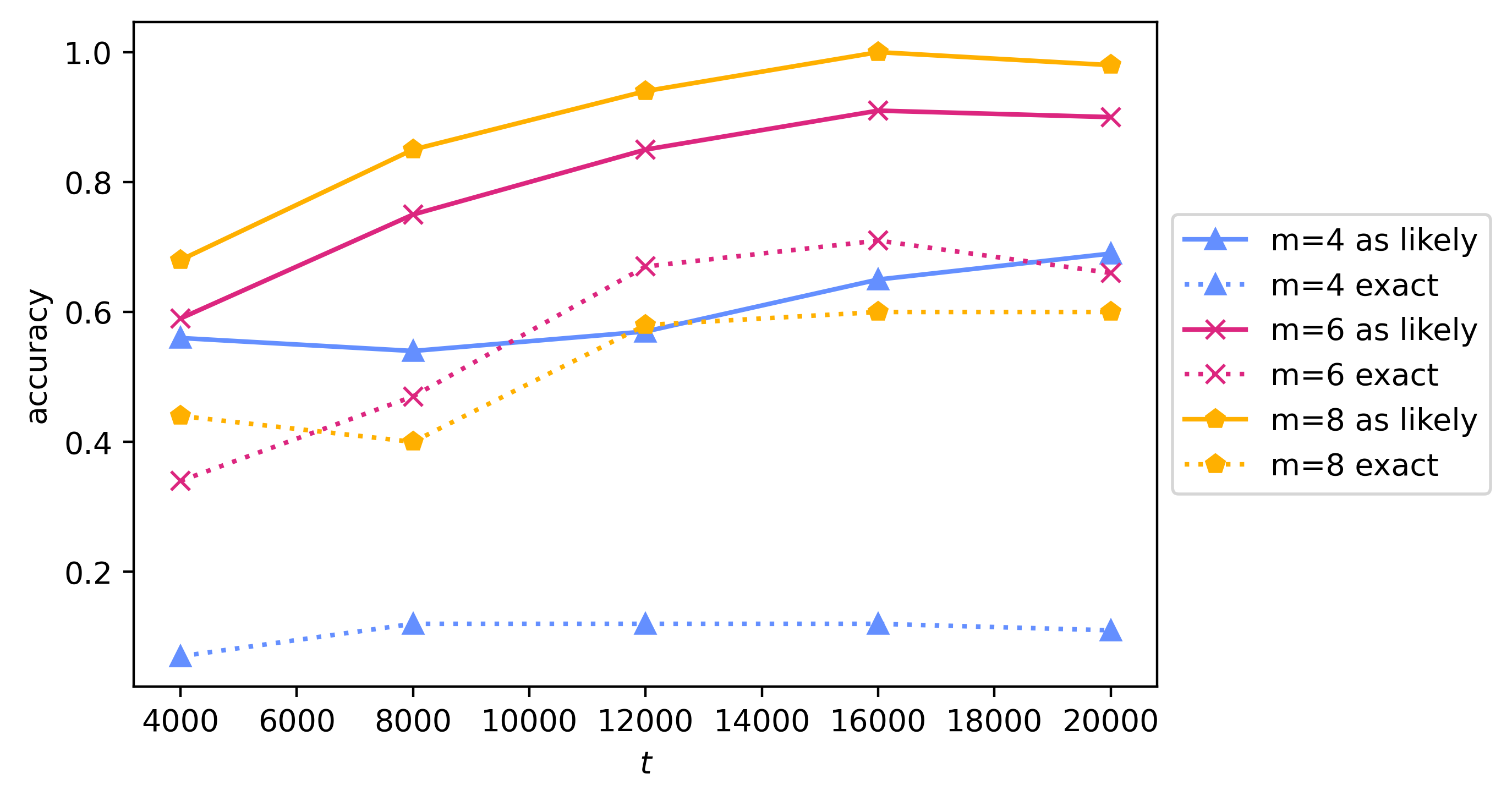}
    \caption{Decoder accuracy with $n=16$, $m=4,6,8$, $p=0.1$, and $q=0.3$ for a varying number of checks $t$ at each partition type.}
    \label{fig:plot_accuracy}
\end{figure}

We also compute the Adjusted Rand Index (ARI), a standard information-theoretic metric for computing the similarity between two clusterings. A completely random clustering should have ARI value equal to 0, while a perfect clustering has ARI equal to 1. See, e.g., \cite{ARI1, ARI2} for a formal definition. 

The results in Figure \ref{fig:plot_ARI} show the mean ARI of the decoder for the same set of planted partitions as in the accuracies plot of Figure \ref{fig:plot_accuracy}. We again see that exact recovery is most likely when $m=6$ and performance declines for $m=8$, which mirrors the results of the ``exact'' performance in Figure~\ref{fig:plot_accuracy}. 

\begin{figure}[hbt!]
    \centering
    \includegraphics[width=0.9\linewidth]{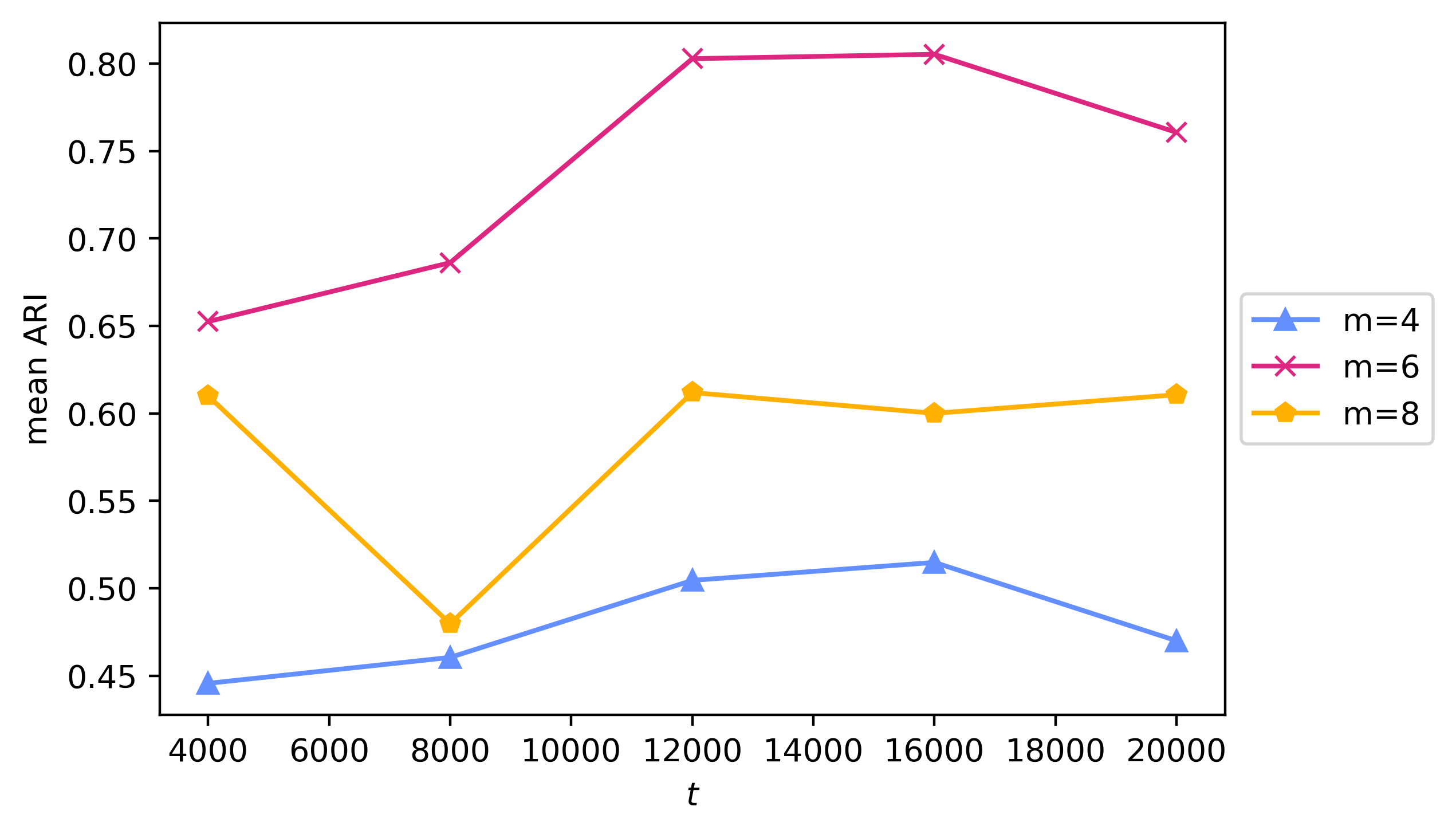}
    \caption{Mean ARI for the decoder with $n=16$, $m=4,6,8$, $p=0.1$, and $q=0.3$ for a varying number of checks $t$ at each partition type.}
    \label{fig:plot_ARI}
\end{figure}

% \jess{ARI for Zachary comment; comparable to Louvain Method, Infomap (citations)}

We would be remiss not to apply our algorithm to a classic real-world example: Zachary's Karate Club \cite{zach_og,zach_debut}.  Many state-of-the-art community detection algorithms are applied to this network in \cite{zach_parameters}, where the authors find ARI ranging from 0.33 to 0.94, depending on the algorithm. (See Table 5 of \cite{zach_parameters}.) We find middling performance for our algorithm on this benchmark; considering the 34-vertex network as an element of $C_{34,10}$, and matching the empirical $p,q$ values for the dataset in the inputs to SASH, we achieve a peak ARI of 0.46 and a mean ARI of 0.37 across 100 trials, while searching within $C_{34,15}$ yields a peak ARI of 0.57 and a mean ARI of 0.41. While these results are competitive with some established techniques, we attribute the lack of compelling performance to the relative sparsity of the real-world network and consider this density question an open line of future work. 

\section{Conclusion}
\label{sec:conc}

We introduced a novel framework for approaching community detection in networks from the perspective of error-correcting codes. We presented the parameters of our community codes, including positioning the minimum discrepancy of a code as the appropriate measure of distance in this setting. We then presented a decoding algorithm, SASH, that leverages the code and channel properties. Finally, we included performance results for the algorithm.

Future directions include understanding the effect of mismatches between actual and approximated $p$, $q$, and $m$, an analysis of the computational complexity of our algorithm, improving SASH to account for variation in allowed partition types, and optimizing SASH for sparse networks.

\bibliographystyle{IEEEtran}
\bibliography{references}

\end{document}